\documentclass[journal,10pt,twocolumn,twoside]{IEEEtran} 
\IEEEoverridecommandlockouts
\usepackage{amsmath}
\usepackage{amsfonts}
\usepackage{cases}
\usepackage{setspace}
\usepackage{fancybox}
\usepackage{subfigure}
\usepackage{epsfig}
\usepackage{graphicx}
\usepackage{epstopdf}
\usepackage{float}
\usepackage{multirow}
\usepackage{color}%
\usepackage{amsmath}
\usepackage{multirow}

\usepackage{indentfirst}
\usepackage{dsfont}
\usepackage{amsfonts}
\usepackage{times,amsmath,color,amssymb,epsfig,cite,subfigure,algorithm,algorithmic}

\newenvironment{proof}[1][Proof]{\begin{trivlist}
\item[\hskip \labelsep {\bfseries #1}]}{\end{trivlist}}

\newenvironment{fact}[1][Fact]{\begin{trivlist}
\item[\hskip \labelsep {\bfseries #1}]}{\end{trivlist}}

\newcommand{\qed}{\nobreak \ifvmode \relax \else
      \ifdim\lastskip<1.5em \hskip-\lastskip
      \hskip1.5em plus0em minus0.5em \fi \nobreak
      \vrule height0.40em width0.6em depth0.25em\fi}

\begin{document}
%
\title{Minimum-Polytope-Based Linear Programming Decoder for LDPC Codes via ADMM Approach}

\author{Jing Bai, 
        Yongchao Wang, \emph{Member, IEEE}, 
        Francis C. M. Lau, \emph{Senior Member, IEEE}}

%

\markboth{}%
{}

\maketitle

\begin{abstract}
In this letter, we develop an efficient linear programming (LP) decoding algorithm for low-density parity-check (LDPC) codes.
We first relax the maximum likelihood (ML) decoding problem to a LP problem by using check-node decomposition.
Then, to solve the resulting LP problem, we propose an efficient iterative algorithm based on the alternating direction method of multipliers (ADMM) technique.
In addition, the feasibility analysis of the proposed algorithm is presented.
Furthermore, through exploiting the sparsity and orthogonality structures of the LP problem, the computational complexity of the proposed decoding algorithm increases linearly
with the length of the LDPC code.
Simulation results demonstrate that our proposed algorithm achieves better performance than other competing algorithms in terms of decoding time.
\end{abstract}


\IEEEpeerreviewmaketitle

\section{Introduction}
Linear programming (LP) decoding has been considered as a promising decoding approach for low-density parity-check
(LDPC) codes \cite{FeldmanLP}.
Compared with the classical belief propagation (BP) decoder, LP decoding can be supported by theoretical guarantees for decoding performance.
However, the original LP decoder suffers from high complexity and thus cannot be
efficiently implemented in practice.
Therefore, reducing the complexity of LP decoding becomes an interesting but challenging topic.
There exist two popular research directions: formulating new small-scale LP decoding problems and designing more efficient problem-solving algorithms.

Considering the first direction, several studies have focused on transforming the maximum-likelihood (ML) decoding problem into LP ones with fewer constraints.
For instance, the authors in \cite{adaptive-LP} proposed a small-sized LP decoding model by adaptively adding necessary constraints.
In \cite{check-node-decompetition1}, another new LP formulation,
in which the variables and constraints grow only linearly with the check node degree,
was introduced.

As for the second research direction, increasing research efforts have been devoted to investigating how the inherent structures of LP problems be exploited to develop more efficient decoding algorithms.
In \cite{Barman-ADMM}, a distributed algorithm based on the alternating direction method of multipliers (ADMM) technique \cite{Boyd} was first proposed to solve the original LP decoding problem \cite{FeldmanLP}.
However, this ADMM-based algorithm involves complex check-polytope projection and thus limits its efficiency.
Subsequently, the authors in \cite{efficient-projection1,efficient-projection2,efficient-projection3} independently optimized the projection algorithm \cite{Barman-ADMM} to further reduce the complexity of the LP decoder \cite{Barman-ADMM}.
In \cite{jiao-zhang}, a projection-reduction method was proposed by decreasing the number of projection operations to simplify the ADMM decoding algorithm.

To the best of our knowledge, developing efficient algorithms for other types of LP decoders has not been well investigated in the existing literatures.
In this letter, we consider the minimum-polytope-based LP (MPB-LP) decoder which is formulated by check-node decomposition, and develop an efficient decoding algorithm based on the ADMM technique.
The main contributions of this work are threefold.
\begin{itemize}
 \item We show the sparsity and orthogonality properties of the MPB-LP problem. By exploiting these structures, we propose an ADMM-based decoding algorithm in which each updating step can be solved efficiently.
 \item We analyze the feasibility of the proposed ADMM-based algorithm and further show that the complexity of our proposed algorithm increases linearly with the length of LDPC codes.
 \item Simulation results present that the proposed ADMM-based LP decoder consumes less decoding time than competing LP decoders, and displays superior error-rate performance than BP decoder at high signal-to-noise (SNR) regions with a shorter decoding time.
\end{itemize}

\section{Minimum-polytope-based LP decoding model}\label{LP fomulation_Pre}

\subsection{General MPB-LP decoding model}\label{problem-formulation}
Consider a binary LDPC code $\mathcal{C}$ defined by an $m \times n$ parity-check matrix $\mathbf{H} = [H_{ji}]_{m \times n}$.
Let $\mathcal{I} = \{1,\ldots,n\}$ and $\mathcal{J} = \{1,\ldots,m\}$ denote the set of variable nodes and check nodes of $\mathcal{C}$ respectively.
Let the degree of a check node be the number of its neighboring variable nodes.
Denote the minimum polytope $\mathcal{P}_{3}$ \cite{check-node-decompetition 1} describing the convex hull of the parity-check constraint of a degree-3 check node by
\begin{equation}\label{degree-3 minimum polytope}
  \begin{split}
    &\hspace{0cm} \mathcal{P}_{3} = \{(x_{1},x_{2},x_{3}):x_{1} + x_{2} + x_{3} \leq 2,   \\
    &\hspace{1.6cm} ~x_{1}- x_{2}-x_{3} \leq  0, ~ - x_{1}+x_{2}- x_{3} \leq 0,  \\
    &\hspace{1.2cm} ~ - x_{1}-x_{2}+ x_{3} \leq 0, ~~~ 0\leq x_{1},x_{2},x_{3} \leq 1 \}.
  \end{split}
\end{equation}

Now suppose that a codeword $\mathbf{x} \in \mathcal{C}$ is transmitted over a noisy memoryless binary-input output-symmetric channel, and a vector $\mathbf{r}$ is received.
By decomposing each high-degree check node into a number of degree-$3$ check nodes with a certain number auxiliary variables, the maximum likelihood (ML) decoding problem can be relaxed to the following MPB-LP form \cite{check-node-decompetition1}
{\setlength\abovedisplayskip{3pt}
\setlength\belowdisplayskip{3pt}
\setlength\jot{1pt}
\begin{equation}\label{cascaded LP}
 \begin{split}
   &\mathop {\min }\limits_\mathbf{x} \hspace{0.3cm} \pmb{\gamma}^{T}\mathbf{x}, \hspace{0.5cm} \textrm{s.t.} \hspace{0.3cm} \mathbf{x} \in \bigcap_{j=1}^{m} \Phi_{j},\ \ j\in\mathcal{J},
 \end{split}
\end{equation}
where} $\pmb{\gamma}$ is the vector of log-likelihood ratios (LLR) defined by $\gamma_{i} = {\rm{log}}{\left(\frac{{\rm{Pr}}(r_{i}|x_{i}=0)}{{\rm{Pr}}(r_{i}|x_{i}=1)}\right)}$, and $\Phi_{j}$ denotes the intersection of all the minimum polytopes for the $j$th parity-check constraint.

\subsection{Standard MPB-LP decoding model}
To make the above MPB-LP problem \eqref{cascaded LP} easier to analyze, we
convert it into a standard LP form.
We denote $d_j$ as the degree of the $j$-th check node in $\mathcal{C}$.
Let $\Gamma_{c}$ and $\Gamma_a$ be the number of minimum polytopes and
the number of auxiliary variables, respectively.
Then $\Gamma_{c}= \sum_{j=1}^{m}(d_j-2)$ and $\Gamma_{a} = \sum_{j=1}^{m}(d_j-3)$  \cite{check-node-decompetition1}.
We further define $\mathbf{q} = \left[ \begin{array}[pos]{l}
\pmb{\gamma} \\
 \mathbf{0}_{\Gamma_a} \end{array} \right]$ and
$\mathbf{v} =
 \left[ \begin{array}[pos]{l}
\mathbf{x} \\
 \mathbf{u}_{\Gamma_a} \end{array} \right]$
where $\mathbf{0}_{\Gamma_a}$ is the length-$\Gamma_a$
all-zero vector and  $\mathbf{u}_{\Gamma_a} \in [0,1]^{\Gamma_a}$
represents the auxiliary variable vector.

Since each minimum polytope is characterized by three variable nodes,
 we assume that the $\tau$-th minimum polytope
($\tau = 1, 2, \ldots,\Gamma_c$)
is related to the $v_{\tau_{1}}$-th, $v_{\tau_{2}}$-th, and $v_{\tau_{3}}$-th variables in $\mathbf{v}$.
For the $\tau$-th minimum polytope, we first construct a corresponding
matrix $\mathbf{Q}_{\tau}\in\{0,1\}^{3 \times (n+\Gamma_a)}$, in which each row contains
only one nonzero element and the indices of the three nonzero elements are
given by
$(1, v_{\tau_{1}}), (2, v_{\tau_{2}})$, and $(3, v_{\tau_{3}})$.
Based on the first four inequalities in \eqref{degree-3 minimum polytope}, we also construct a vector $\mathbf{t}=[2\; 0\; 0\; 0]^T $ and a matrix
\begin{equation}\label{F matrix}
  \mathbf{F} = \begin{bmatrix}
   \begin{array}{ccc}
    ~~1 & ~~ 1 & ~~1 \\
    ~~1 & -1  & -1 \\
   -1 & ~~ 1 & -1 \\
   -1 & -1 & ~~1 \\
   \end{array}
  \end{bmatrix},
\end{equation}
where $(\cdot)^T$ represents the transpose operation.
Then the $\tau$-th minimum polytope ($\tau = 1, 2, \ldots,\Gamma_c$) can be expressed as
\begin{equation}\label{minimum polytope}
  \mathbf{F} \mathbf{Q}_{\tau}\mathbf{v}\preceq \mathbf{t}.
\end{equation}
We further define a matrix $\mathbf{A}=[\mathbf{F} \mathbf{Q}_1;...;\mathbf{F} \mathbf{Q}_\tau;...;\mathbf{F} \mathbf{Q}_{\Gamma_c}]$, which is cascaded by each $\mathbf{F} \mathbf{Q}_\tau$, and a vector $\mathbf{b}=\mathbf{1}_{{\Gamma_c}} \otimes \mathbf{t}$ where $\otimes$ denotes the Kronecker product and
$\mathbf{1}_{{\Gamma_c}}$ is the length-$\Gamma_c$ all-one vector.
Letting $M=4\Gamma_c$ and $N = n+\Gamma_a$, then $\mathbf{A}\in\mathbb{R}^{M \times N}$
and $\mathbf{b} \! \in \! \mathbb{R}^{M}$.
Subsequently, we reformulate the MPB-LP problem \eqref{cascaded LP} equivalently to the following linear program
\begin{equation}\label{C-LP model}
  \begin{split}
    &\hspace{0.0cm} \underset{\mathbf{v}}\min \hspace{0.2cm} \mathbf{q}^{T}\mathbf{v}, \hspace{0.3cm} \textrm{s.t.} \hspace{0.2cm} \mathbf{Av} \preceq \mathbf{b}, \hspace{0.2cm} \mathbf{v} \in [0,1]^N.
  \end{split}
\end{equation}

The problem \eqref{C-LP model} certainly can be solved by general-purpose LP solvers, such as the SIMPLEX method or the interior-point method.
However, the high computational complexity limits their practical applications.
Therefore, in the following, we will propose an efficient algorithm based on the ADMM technique to solve the problem \eqref{C-LP model}.

\section{Proposed ADMM-Based Decoding Algorithm}
This section discusses the detailed ADMM-based algorithm for solving the MPB-LP decoding problem \eqref{C-LP model}, and presents its feasibility and complexity analysis.

\subsection{ADMM algorithm framework}
To make the problem \eqref{C-LP model} fitting into the ADMM template, we have to add an auxiliary variable vector $\mathbf{w}\in\mathbb{R}_{+}^{M}$ so as to change the inequality into an equality constraint.
With $\mathbf{w}$, the linear program \eqref{C-LP model} is equivalent to
{\setlength\abovedisplayskip{3pt}
\setlength\belowdisplayskip{3pt}
\setlength\jot{1pt}
\begin{equation}\label{ADMM model}
 \begin{split}
   &\hspace{0.0cm} \mathop {\min }\limits_\mathbf{v} \hspace{0.2cm} \mathbf{q}^{T}\mathbf{v},              \\
   &\hspace{0.2cm}\textrm{s.t.} \hspace{0.30cm} \mathbf{Av} + \mathbf{w} = \mathbf{b},\\
   &\hspace{0.95cm}  \mathbf{v}\in [0,1]^{N}, \hspace{0.3cm} \mathbf{w} \in \mathbb{R}_{+}^{M}.
 \end{split}
\end{equation}}

The augmented Lagrangian function (using the scaled dual variable) for  problem \eqref{ADMM model}
 is expressed by
\begin{equation}\label{augmented-Lagrangian-LP}
  \emph{L}_{\mu}(\mathbf{v},\mathbf{w},\pmb{\lambda}) = \mathbf{q}^{T}\mathbf{v} + \frac{\mu}{2} \| \mathbf{Av+w-b}+\pmb{\lambda}\|_2^{2}-\frac{\mu}{2} \|\pmb{\lambda}\|_2^{2},
\end{equation}
where $\pmb{\lambda} \in \mathbb{R}^{M}$ denotes the scaled dual variable vector, $\mu > 0$ is the penalty parameter, and $\lVert\cdot\rVert_2$ represents the 2-norm operator.
Then a typical ADMM-based algorithm for solving  \eqref{ADMM model} can be described by the following iterations
{\setlength\abovedisplayskip{3pt}
\setlength\belowdisplayskip{3pt}
\setlength\jot{1pt}
\begin{subequations}\label{ADMM-update-lp}
  \begin{align}
   & \mathbf{v}^{(k+1)} = \mathop{\textrm{argmin}}\limits_{\mathbf{v}\in [0,1]^{N}} ~ \emph{L}_{\mu}(\mathbf{v},\mathbf{w}^{(k)},\pmb{\lambda}^{(k)}), \label{v-update-ori}  \\
   & \mathbf{w}^{(k+1)} = \mathop{\textrm{argmin}}\limits_{\mathbf{w} \in \mathbb{R}_{+}^{M}} ~ \emph{L}_{\mu}(\mathbf{v}^{(k+1)},\mathbf{w},\pmb{\lambda}^{(k)}), \label{w-update-ori} \\
   & \pmb{\lambda}^{(k+1)} = \pmb{\lambda}^{(k)} +(\mathbf{Av}^{(k+1)}+\mathbf{w}^{(k+1)}-\mathbf{b}), \label{lamda-update-ori}
 \end{align}
\end{subequations}
where} $k$ denotes the iteration number.

Observing \eqref{ADMM-update-lp}, the majority of its computational complexity depends on solving \eqref{v-update-ori} and \eqref{w-update-ori}.
However, both of them can be implemented efficiently since matrix $\mathbf{A}$ owns the following properties:
\begin{fact}
  The matrix $\mathbf{A}$ possesses the following properties.
  \begin{itemize}
    \item The elements of $\mathbf{A}$ (i.e., $A_{ji}$) are either $0$, $-$1, or $1$.
    \item $\mathbf{A}$ is sparse.
    \item The column vectors of $\mathbf{A}$ are orthogonal to one another.
  \end{itemize}
\end{fact}

\begin{proof}
See Appendix \ref{fact_proof}.
\end{proof}

In the following we will show that each sub-problem in the above algorithm \eqref{ADMM-update-lp} can be solved efficiently by exploiting these favourable features of matrix $\mathbf{A}$.

\subsection{Solving the sub-problem \eqref{v-update-ori}}

Since matrix $\mathbf{A}$ is column orthogonal, $\mathbf{A}^{T}\mathbf{A}$ is a diagonal matrix. This implies that variables in the problem \eqref{v-update-ori} are separable.
Therefore, solving the problem \eqref{v-update-ori} can be equivalent to solving the following $N$ subproblems independently
{\setlength\abovedisplayskip{3pt}
\setlength\belowdisplayskip{3pt}
\setlength\jot{1pt}
\begin{subequations}\label{vsp}
  \begin{align}
   & \underset{\mathbf{v}}\min \hspace{0.1cm} \frac{1}{2}(\mu e_i)v_i^2 + \Big(q_i\!+\!\mu \hat{\mathbf{a}}_i^T\big(\mathbf{w}^{(k)}\!-\!\mathbf{b}\!+\!\pmb{\lambda}^{(k)}\big)\Big)v_i,  \label{vsp a}\\
   &\hspace{0.1cm} \textrm{s.t.}\hspace{0.2cm} v_i\in [0,1], \label{vsp b}
  \end{align}
\end{subequations}
where} $\mathbf{e}={\rm diag}(\mathbf{A}^T\mathbf{A})=[e_{1},\cdots,e_{n+\Gamma_a}]^{T}$ and $\rm{ diag}(\cdot)$ denotes the operator of extracting the diagonal vector of a matrix, and $\hat{\mathbf{a}}_i$ is the $i$-th column of the matrix $\mathbf{A}$.
Then, the procedures for solving the sub-problem \eqref{vsp} can be summarized as follows: setting the gradient of the objective \eqref{vsp a} to zero, then projecting the resulting solution to the interval $[0,1]$. Finally, we can obtain the the following solution of the problem \eqref{vsp} as
{\setlength\abovedisplayskip{3pt}
\setlength\belowdisplayskip{3pt}
\setlength\jot{1pt}
\begin{equation}\label{vi-update}
v_{i}^{(k+1)} = \Pi_{[0,1]}\bigg(\frac{1}{e_i} \Big( \hat{\mathbf{a}}_i^T \big(\mathbf{b}- \mathbf{w}^{(k)}-\pmb{\lambda}^{(k)}\big)-\frac{q_{i}}{\mu}\Big)\bigg),
\end{equation}
where} $\Pi_{[0,1]}(\cdot)$ denotes the Euclidean projection operator onto the interval $[0,1]$.

\subsection{Solving the sub-problem \eqref{w-update-ori}}
Observing \eqref{w-update-ori}, we can find that the variables in $\mathbf{w}$ are also separable in either objective or constraints.
Hence, the sub-problem \eqref{w-update-ori} can be solved by fixing $\mathbf{v}$ and $\pmb{\lambda}$, and then minimizing $\emph{L}_{\mu}(\mathbf{v},\mathbf{w},\pmb{\lambda})$ under the constraint $\mathbf{w} \in \mathbb{R}_{+}^{M}$ (i.e., $\mathbf{w} \in [0,+\infty]^{M}$).
Using a technique similar to that described in the previous section,
$\mathbf{w}^{(k+1)}$ can be updated by
{\setlength\abovedisplayskip{3pt}
\setlength\belowdisplayskip{3pt}
\setlength\jot{1pt}
\begin{equation}\label{w-update-lp}
   \mathbf{w}^{(k+1)} = \Pi_{[0,+\infty]^{M}}\Big(\mathbf{b}-\mathbf{Av}^{(k+1)}-\pmb{\lambda}^{(k)}\Big),  \\
\end{equation}
where} $\Pi_{[0,+\infty]^{M}}(\cdot)$ denotes the Euclidean projection onto the positive quadrants $[0,+\infty]^{M}$.
Clearly, the $\mathbf{w}$-update can also be written in a component-wise manner,
i.e.,
{\setlength\abovedisplayskip{3pt}
\setlength\belowdisplayskip{3pt}
\setlength\jot{1pt}
\begin{equation}\label{w_component_update_lp}
   w_{j}^{(k+1)} = \Pi_{[0,+\infty]}\Big(b_{j}-\mathbf{a}_{j}^{T}\mathbf{v}^{(k+1)}-\lambda_{j}^{(k)}\Big),  \\
\end{equation}}
where $\mathbf{a}_j^{T}$ is the $j$-th row of matrix $\mathbf{A}$ and $j\in \{1,2,\ldots,M\}$.

Finally, we summarize the proposed minimum-polytope-based ADMM-solved
(MPB-ADMM) LP decoding method in \emph{Algorithm \ref{ADMM-LP-alg}}.
Here we make three remarks for the proposed algorithm:
\begin{enumerate}[(1)]
  \item In \emph{Algorithm \ref{ADMM-LP-alg}}, the multiplications with respect to matrix $\mathbf{A}$, such as $\hat{\mathbf{a}}_i^T \mathbf{b}$ \emph{et al}., are just considered as addition operations based on the first property of matrix $\mathbf{A}$.
  \item All the variables in $\mathbf{v}$ can be updated in parallel.
  \item Compared with previous works \cite{Barman-ADMM,efficient-projection1,efficient-projection2,efficient-projection3,jiao-zhang} that require complex check-polytope projections, each variable update \eqref{w_component_update_lp} in $\mathbf{w}$ only involves a simple Euclidean projection onto the positive quadrant.
      Hence the $\mathbf{w}$-update significantly reduces the decoding complexity of the proposed ADMM algorithm and can be implemented much more easily in practice.
\end{enumerate}

\subsection{Feasibility analysis}

The convergence analysis of the our proposed algorithm can be found in \cite{Boyd}.
Since the LP problem \eqref{C-LP model} is convex, our proposed algorithm is guaranteed to converge to the global solution $\mathbf{v}^{*}$.
On the other hand, as shown in \cite{check-node-decompetition1}, the MPB-LP decoder \eqref{C-LP model} is equivalent to the original LP decoder.
Therefore, our proposed decoder shares the same properties of
the original LP decoder such as \emph{all-zeros assumption} and \emph{ML-certificate property} \cite{FeldmanLP}.
As a result, if the output of \emph{Algorithm \ref{ADMM-LP-alg}} is an integral solution denoted by $\mathbf{v}^{*}$, the \emph{ML-certificate property} ensures that $\mathbf{x}^{*}$ corresponds to the ML solution.

\subsection{Complexity analysis}

Since the matrix $\mathbf{A}$ contains only $0, -1$, or $1$ elements, its related matrix multiplications can all be performed using addition operations.
Thus, by  taking only the multiplications related to $1/{e}_{i}$ into account,
the computational complexity of updating $\mathbf{v}^{(k+1)}$ in Algorithm \ref{ADMM-LP-alg} at each iteration is $\mathcal{O}(N)$.
By exploiting the desirable properties of matrix $\mathbf{A}$,
the computations of $\mathbf{w}^{(k+1)}$ and $\pmb{\lambda}^{(k+1)}$
can be completed using only addition operations.
Hence their computational complexities are both $\mathcal{O}(1)$.
Combining all the above analysis, the total computational complexity of our proposed algorithm in each iteration is $\mathcal{O}(N)$, i.e., $\mathcal{O}(n+\Gamma_{a})$.
Moreover, $\Gamma_{a}$ is comparable to the code length since the parity check matrix $\mathbf{H}$ is sparse in the case of LDPC codes.
Thus, the complexity of the proposed ADMM-based algorithm in each iteration is linear to the length of LDPC codes.

In \emph{Table \ref{Complexity-Comparison}}, we compare the complexities of our proposed algorithm and the check-polytope-based ADMM-solved (CPB-ADMM) LP decoding algorithms in \cite{efficient-projection1} and \cite{efficient-projection3}.
We can observe that our proposed MPB-ADMM LP decoder achieves
the lowest computational complexity.

\begin{algorithm}[t]
\caption{{\color{black} Proposed MPB-ADMM LP decoding algorithm}} 
\label{ADMM-LP-alg}
\begin{algorithmic}[1]
\STATE Calculate the log-likelihood ratio $\pmb{\gamma}$ based on received vector $\mathbf{r}$ and construct the vector $\mathbf{q}$ based on $\pmb{\gamma}$. \\
\STATE Construct the $M \times N$ matrix $\mathbf{A}$ based on the parity-check matrix $\mathbf{H}$ and construct the vector $\mathbf{b}$.
\STATE  Set $\mathbf{w}^{(0)}$ and $\pmb{\lambda}^{(0)}$ to the length-$M$ all-zero vector, and initialize the iteration number $k=0$. \\
\STATE      \textbf{repeat}  \\
\STATE~~\textbf{for} $i \in \{1, 2, \ldots, N\}$ \textbf{do} \\ 
\STATE~~~~~Update $v_{i}^{(k+1)} \!=\!\Pi_{[0,1]}\bigg(\frac{1}{e_i} \Big( \hat{\mathbf{a}}_i^T \big(\mathbf{b}- \mathbf{w}^{(k)}-\pmb{\lambda}^{(k)}\big)\!-\!\frac{q_{i}}{\mu}\Big)\bigg)$. \\
\STATE~~\textbf{end for} \\
\STATE~~\textbf{for} $j \in \{1, 2, \ldots, M\}$ \textbf{do}   \\ 
\STATE~~~~~Update $w_{j}^{(k+1)} = \Pi_{[0,+\infty]}\Big(b_{j}-\mathbf{a}_{j}^{T}\mathbf{v}^{(k+1)}-\lambda_{j}^{(k)}\Big)$. \\
\STATE~~~~~Update $\lambda_{j}^{(k+1)} = \lambda_{j}^{(k)} + (\mathbf{a}_{j}^{T}\mathbf{v}^{(k+1)}+w_{j}^{(k+1)}-b_{j})$. \\
\STATE~~\textbf{end for} \\
\STATE\textbf{until} $\parallel \mathbf{Av}^{(k+1)}+\mathbf{w}^{(k+1)}-\mathbf{b} \parallel_{2}^{2} \ \leq \ \xi$ \\
~~~~~~ and $\parallel \mathbf{w}^{(k+1)}-\mathbf{w}^{(k)} \parallel_{2}^{2} \ \leq \ \xi$.
\end{algorithmic}
\end{algorithm}

\begin{table*}\footnotesize
\caption{Complexity of Different LP Decoding Algorithms Per Iteration. $d$ denotes the largest check-node degree. $I_{max}$ is the  maximum number of iterations set in the projection algorithm of \cite{efficient-projection3}.}
\label{Complexity-Comparison}
\begin{center}
\begin{tabular}{|c|c|c|c|c|c|}
\hline
\hline
\textbf{Decoding Algorithm}& \textrm{Original variables}&\textrm{Auxiliary variables}&\textrm{Dual variables} & Overall Complexity\\\hline
Proposed MPB-ADMM LP & $\mathcal{O}(n+m(d-3))$ & $\mathcal{O}(1)$ & $\mathcal{O}(1)$ & $\mathcal{O}(n+m(d-3))$\\\hline
CPB-ADMM LP \cite{efficient-projection1} & $\mathcal{O}(n)$ & $\mathcal{O}(md \log d)$ & $\mathcal{O}(1)$ & $\mathcal{O}(n+md \log d)$\\\hline
CPB-ADMM LP \cite{efficient-projection3} & $\mathcal{O}(n)$ & $\mathcal{O}(mI_{max})$ & $\mathcal{O}(1)$ & $\mathcal{O}(n+mI_{max})$\\\hline
 \end{tabular}
 \end{center}
\end{table*}

\section{Simulation Results}\label{simulation}
In this section, we present the simulation results of our proposed MPB-ADMM LP decoder, the CPB-ADMM LP decoders in \cite{efficient-projection1} and \cite{efficient-projection3}, and the classic sum-product BP decoder.
The simulations are performed on a computer with i5-3470 3.2GHz CPU and 16 GB RAM under the Microsoft Visual C++ 6.0 environment.

We simulate two widely-used LDPC codes, namely the $(2640,1320)$ rate-1/2 (3,6)-regular ``Margulis'' LDPC code $\mathcal{C}_{1}$ \cite{Mackay-code} and the $(576,288)$ rate-1/2 irregular LDPC code $\mathcal{C}_{2}$ in the 802.16e standard \cite{802.16-code}.
All the code bits $\mathbf{x}$ are sent over an additive white Gaussian noise (AWGN) channel using binary phase-shift-keying (BPSK).
In our decoding algorithm,
the penalty parameter $\mu$ is set to $0.6$ and $0.8$ for $\mathcal{C}_{1}$ and $\mathcal{C}_{2}$, respectively\footnote{The parameter $\mu$ is chosen reasonably by plotting the error-rate performance and decoding time as a function of $\mu$ respectively. Due to space limitation, we do not give these detailed figures.}.
The maximum number of decoding iterations is set to $500$ and the tolerance $\xi$ is set to $10^{-5}$ \footnote{The maximum number of iterations 500 and tolerance $10^{-5}$ is considered good enough to achieve the desired tradeoff between error rate and complexity.}.

\begin{figure}[htbp]
  \centering
  \centerline{\psfig{figure=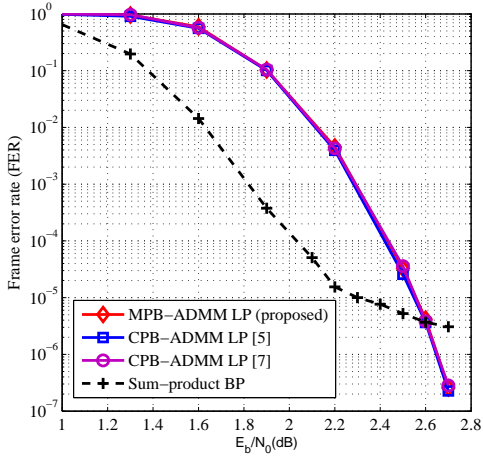,width=7cm,height=6.5cm}}
  \caption{FER performance of $\mathcal{C}_{1}$ using different decoders. $\mathcal{C}_{1}$ denotes the $(2640,1320)$ regular ``Margulis'' LDPC code.}
  \label{fer2640}
\end{figure}

\begin{figure}[htbp]
  \centering
  \centerline{\psfig{figure=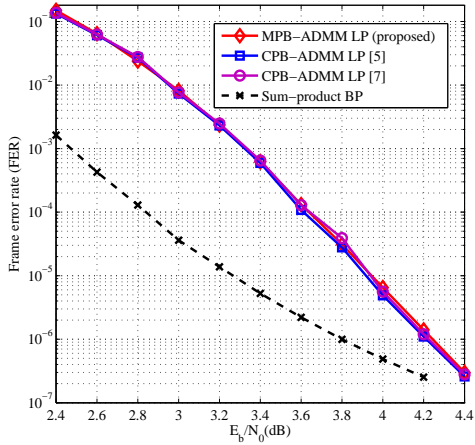,width=7cm,height=6.5cm}}
  \caption{FER performance of $\mathcal{C}_{2}$ using different decoders. $\mathcal{C}_{2}$ denotes the $(576,288)$ irregular LDPC code in the 802.16e standard.}
  \label{fer576}
\end{figure}

\begin{figure}[htbp]
  \centering
  \centerline{\psfig{figure=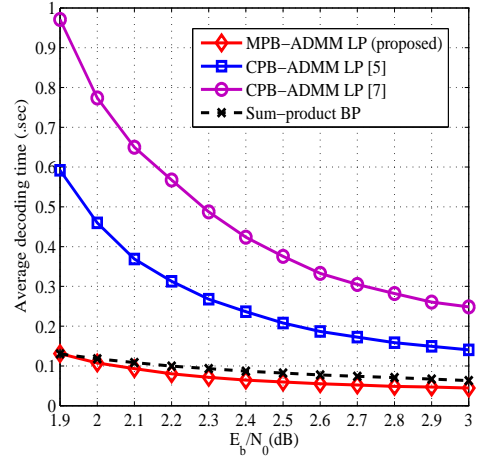,width=7cm,height=6.5cm}}
  \caption{Decoding time of $\mathcal{C}_{1}$ using different decoders. $\mathcal{C}_{1}$ denotes the $(2640,1320)$ regular ``Margulis'' LDPC code.}
  \label{time2640}
\end{figure}

\begin{figure}[htbp]
  \centering
  \centerline{\psfig{figure=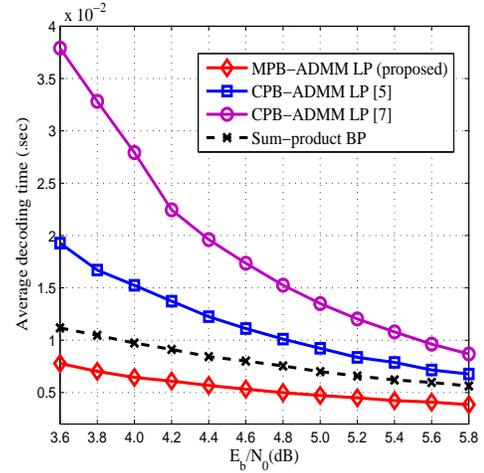,width=7cm,height=6.5cm}}
  \caption{Decoding time of $\mathcal{C}_{2}$ using different decoders. $\mathcal{C}_{2}$ denotes the $(576,288)$ irregular LDPC code in the 802.16e standard.}
  \label{time576}
\end{figure}

Fig.~\ref{fer2640} and Fig.~\ref{fer576} plot the frame-error-rate (FER) performances of $\mathcal{C}_{1}$ and $\mathcal{C}_{2}$, respectively, with different decoding algorithms.
In the case of the MPB/CPB-ADMM LP decoding algorithms, we collect 50 error frames at $E_{b}/N_{0} = 2.7$ dB in Fig.~\ref{fer2640} and $E_{b}/N_{0} = 4.4$ dB in Fig.~\ref{fer576} and 200 for all other SNRs.
The results indicate that our proposed MPB-ADMM LP decoder has the same performance
as the CPB-ADMM LP decoders in  \cite{efficient-projection1} and \cite{efficient-projection3}. Moreover, our proposed decoder outperforms the sum-product BP decoder
at high SNR regions without showing any error floor.

In Fig.~\ref{time2640} and Fig.~\ref{time576}, we make a comparison in the average decoding time of our proposed algorithm and other competing decoding algorithms at different SNRs for $\mathcal{C}_{1}$ and $\mathcal{C}_{2}$, respectively.
The decoding times in these two figures are averaged over one million frames for each decoder.
And note that over relaxation (OR) and early termination (ET) techniques are both applied in ADMM-based LP decoders.
From Fig.\ref{time2640} and Fig.\ref{time576}, we observe that our proposed MPB-ADMM LP decoder is the most efficient among all the competing decoding algorithms, which is consistent with the complexity analysis of \emph{Table \ref{Complexity-Comparison}}.
To be specific, in Fig.~\ref{time2640} we see that our decoding algorithm roughly saves 31\%, 69\% and 83\% decoding time respectively for $\mathcal{C}_{1}$ when $E_{b}/N_{0}=2.8$ dB in comparison with the sum-product BP decoding and the CPB-ADMM LP algorithms \cite{efficient-projection1} and \cite{efficient-projection3}.
Besides, Fig.~\ref{time576} shows that our proposed decoder roughly reduces 32\%, 46\% and 63\% time respectively for $\mathcal{C}_{2}$ when $E_{b}/N_{0}=5.2$ dB compared with the sum-product BP decoder and the other two CPB-ADMM LP decoders.

\section{Conclusion}
In this letter, we propose an efficient LP decoding algorithm based on ADMM technique for LDPC codes.
By exploiting the sparsity and orthogonality of the formulated MPB-LP model, we solve this resulting LP problem efficiently by the ADMM-based algorithm.
Detailed analysis shows that the complexity of the proposed algorithm at each iteration grows linearly with the length of LDPC codes.
Simulation results confirm the efficiency of the proposed LP decoder.

\appendices
\section{Proof of Fact 1}\label{fact_proof}
\begin{proof} We consider the first property of the matrix $\mathbf{A}$. Since each row in the matrix $\mathbf{Q}_{\tau}$ includes one nonzero element ``1'', any element in the matrix $\mathbf{F}\mathbf{Q}_{\tau}$ is either 0, $-$1, or 1.
The same applies to the matrix $\mathbf{A}$.
Moreover, we can see that there are only 12 nonzero elements in $\mathbf{F}\mathbf{Q}_{\tau}$.
It means that there are only $12\Gamma_c$ nonzero elements in the matrix $\mathbf{A}$.
It is far smaller than the size, $4\Gamma_c\times N$, of the matrix $\mathbf{A}$.
Therefore, we obtain the sparsity of $\mathbf{A}$.
Furthermore, we note that any two column vectors in $\mathbf{F}$ are orthogonal to each other.
Thus, each two column vectors in $\mathbf{F}\mathbf{Q}_{\tau}$ are also orthogonal.
This implies that $\mathbf{A}$ is a matrix with orthogonal columns because $\mathbf{A}$ is formed by cascading the matrices $\{\mathbf{F}\mathbf{Q}_{\tau}: \tau=1,2,\ldots,\Gamma_c\}$. 
\end{proof}

%
%

\ifCLASSOPTIONcaptionsoff
  \newpage
\fi



%
%
%





\begin{thebibliography}{99}

\bibitem{FeldmanLP}
J. Feldman, M. T. Wainwright, and D. R. Karger, ``Using linear progamming to decoding binary linear codes,'' \emph{IEEE Trans. Inf. Theory}, vol.~51, no.~1, pp.~954-972, Jan.~2005.


\bibitem{adaptive-LP}
M. H. Taagavi and P. H. Siegel, ``Adaptive methods for linear programming decoding,'' \emph{IEEE Trans. Inf. Theory}, vol. 54, no. 12, pp. 5396-5410, Dec. 2008.

\bibitem{check-node-decompetition1}
K. Yang, X. Wang, and J. Feldman, ``A new linear programming approach to decoding linear block codes,'' \emph{IEEE Trans. Inf. Theory}, vol. 54, no. 3, pp. 1061-1072, Mar. 2008.



%

\bibitem{Barman-ADMM}
S. Barman, X. Liu, S. C. Draper, and B. Recht, ``Decomposition method for large scale LP decoding,'' \emph{IEEE Trans. Inf. Theory}, vol. 59, no. 12, pp. 7870-7886, Dec. 2013.

\bibitem{efficient-projection1}
X. Zhang and P. H. Siegel, ``Efficient iterative LP decoding of LDPC codes with alternating direction method of multipliers,'' in Proc. \emph{IEEE Int. Symp. Inf. Theory}, Istanbul, Turkey, Jul. 2013, pp. 1501-1505.

\bibitem{efficient-projection2}
G. Zhang and R. Heusdens, and W. B. Kleijn, ``Large scale LP decoding with low complexity,'' \emph{IEEE Commun. Lett}., vol. 17, no. 11, pp. 2152-2155, Jun. 2015.

\bibitem{efficient-projection3}
H. Wei and A. H. Banihashemi, ``An iterative check polytope projection algorithm for ADMM-based LP decoding of LDPC codes,'' \emph{IEEE Commun. Lett}., vol. 22, no. 1, pp. 29-32, Jan. 2018.

\bibitem{jiao-zhang}
H. Wei, X. Jiao, and J. Mu, ``Reduced-complexity linear programming decoding based on ADMM for LDPC codes,'' \emph{IEEE Commun. Lett}., vol. 19, no. 6, pp. 909-912, Jun. 2015.

\bibitem{Boyd}
S. Boyd, N. Parikh, E. Chu, B. Peleato, and J. Eckstein, ``Distributed optimization and statistical learning via the alternating direction method of multipliers,'' \emph{Found. Trends Mach. Learn.}, vol. 3, no. 1, pp. 1-122, Jan. 2011.




\bibitem{Mackay-code}
D. J. C. MacKay, Encyclopedia of Sparse Graph Codes. [Online]. Available: http://www.inference.phy.cam.ac.uk/mackay/codes/data.html.

\bibitem{802.16-code}
\emph{LDPC coding for OFDMA PHY}, IEEE C802.16e-05/0066r3, Jan.~2005.


\end{thebibliography}
\end{document}